\documentclass{article}

\usepackage{amsmath} 
\usepackage{amssymb}
\usepackage{clrscode3e} 
\usepackage[utf8]{inputenc} 
\usepackage[T2A]{fontenc}  
\usepackage{graphicx} 
\usepackage[space]{grffile} 

\usepackage{amsthm} 

\usepackage{enumitem}

\usepackage[hyphens,spaces,obeyspaces]{url}

\usepackage{xfrac}\usepackage{relsize}
\usepackage[russian,english]{babel} 

\usepackage[nottoc]{tocbibind} 

\newtheorem{theorem}{Theorem}

\def\vec#1{\mathbf{#1}} 

\usepackage{geometry}
\newcommand\FDunderbraces [1] []
{
	\begin{scriptsize}
 		\ldots \underbrace{1\ldots \underbrace{1\ldots _2}_{k_1}}_{k_m}
	\end{scriptsize}
}

\newcommand\ONEunderbrace [1] []
{
	\begin{scriptsize}
 		\ldots 0 \underbrace{10\ldots 0 _2}_{k}
	\end{scriptsize}
}

\newcommand\plainA [1] []
{
	a_{{(j_{p-1} \ldots j_1)}_2}
}
\newcommand\aWithTilde [1] []
{
\begin{scriptsize}
	\overset{\sim{~}}{a}_{{(j_{p-1} \ldots j_1)}_2}
\end{scriptsize}
}
\newcommand\sumOfCombinations[1][]
{
	\sum\limits_{i=0}^n \binom{n}{i}
}
\title{Bibliography management:\\\texttt{thebibliography} environment}
\author{Share\LaTeX}
\date{}

\theoremstyle{definition}
\newtheorem{definition}{Definition}[section]

\begin{document}

\author{Valeri Aronov}
\title{Fast Derivatives for Multilinear Polynomials} 
\date{\today{}} 
\maketitle{} 
\pagenumbering{arabic}

\begin{abstract}
	The article considers linear functions of many ($n$) variables - multilinear polynomials (MP) \cite{multilinearPolynomial}. The three-steps evaluation is presented that uses the minimal possible number of floating point operations for non-sparse MP at each step. The minimal number of additions is achieved in the algorithm for fast MP derivatives ($FMPD$) calculation. The cost of evaluating all first and second derivatives approaches to the cost of MP evaluation itself with a growing $n$. The $FMPD$ algorithm structure exhibits similarity to the Fast Fourier Transformation (FFT) algorithm.
\end{abstract}
\tableofcontents{} 

\section {Introduction}  

\paragraph{}
This is a linear function of $n$ variables:

\begin{equation} \label{eq:MP}
     a(\vec{x})=\sum_{i=0}^{2^{n}-1} r_i \cdot \prod_{\substack{k=1\\i=(i_n...i_k...i_1)_{2}}}^{n}x_{k}^{i_{k}}
     \text{, where } \{ \vec{x}, \vec{r} \in \mathbb{R}^n \}
\end{equation}

\paragraph{}
It is called a multilinear polynomial (MP) \cite{multilinearPolynomial}.

\paragraph{}
For example (for n=2):
$$a(\vec{x})=r_{00_{2}}+r_{01_{2}} \cdot x_{1}+r_{10_{2}} \cdot x_{2}+r_{11_{2}} \cdot x_{1} \cdot x_{2}$$

\paragraph{}
The binary notation for $i$ is natural for the fast MP derivatives ($FMPD$) algorithm presented in this article. The transfer function of linear analog electric circuits has MPs of circuit parameters in its numerator and denominator with coefficients being complex polynomials  of frequency  and $\vec{x}$ - values of the circuit parameters. Some software systems for linear circuit analysis generate and use symbolic (analytical) parameters presentation of the transfer functions for fast multiple transfer function evaluations \cite{symbolicTransferFunctions}. It makes it possible to build a symbolic presentation for for derivatives of transfer functions as well. $FMPD$ algorithm was first introduced in \cite{avtoreferat} for the linear circuit parameters optimization.

\section {Algorithm for Multilinear Polynomial Derivatives Calculation}

\paragraph{}

This article uses ${q_i(\vec{x})}$ designation:

\begin{equation} \label{eq:Qs}
	q_i(\vec{x})=r_i \cdot \prod_{\substack{k=1\\i=(i_n...i_k...i_1)_{2}}}^{n}x_{k}^{i_{k}}
\end{equation}

Their derivatives then are:
	$${({q_i(\vec{x})})}_{x_k}^{'}=q_{i}(\vec{x})/x_k$$
for all $i_k=1$. $q_i(\vec{x})$ does not depend on those $x_k$ which have $i_k=0$ and corresponding derivatives are equal $0$.

The following definitions of full and partial sums of $q_i(\vec{x})$ are used here:
	$$a_0=a(\vec{x})$$
	$$a_{0\ldots01_2}=\sum_{\substack{i=0\\ i=(i_n \ldots i_{2}1)_2}}^{2^{n}-1} q_i$$
	$$a_{0\ldots10_2}=\sum_{\substack{i=0\\i=(i_n \ldots i_{3}1 i_1)_2}}^{2^{n}-1} q_i$$

For two $n$-bit integers $j = (j_1, \ldots, j_n)_2$ and $i = (i_1, \ldots, i_n)_2$ we write 
$j \succ i$ if their binary digits satisfy $j_\nu \ge i_\nu$, $\nu =1, \ldots, n$.
Then 
	\begin{equation} \label{eq:PartialSums}
		a_i = \sum_{j \succ i}  q_j
	\end{equation}
	


	$$a_{1 \dots 1} = q_{1 \dots 1}$$

Given integers $1  \le k_1< \ldots< k_m \le n$ we denote by $t(k_1, \ldots, k_m)$ the 
$n$-bit integer having 1 at positions $k_1, \ldots, k_m$ and $0$ at all other positions. Then MP derivatives are:

	\begin{equation} \label{eq:As}
	     {({a(\vec{x})})}_{x_{k_1} \ldots x_{k_m}}^{(m)}=a_{t(k_1, \ldots, k_m)}(\vec{x})/\prod_{j=1}^{m} x_{k_j}
	\end{equation}


The task of MP derivatives calculation is: given $\vec{r}$ and $\vec{x}$ and using MP presentation \eqref{eq:MP} calculate all derivatives as in \eqref{eq:As} including the evaluation of MP itself. MPs in this article are considered to be 'non-sparse'. It means that either $r_i\neq0$ for all $i$ or the savings in MP and its derivatives calculation due to omission of $r_i=0$ are negligible.

Execute the task in three steps:
\begin{enumerate}
	\item Calculate all items as in $\eqref{eq:Qs}$
	\item Calculate all $a_i$ as in $\eqref{eq:PartialSums}$
	\item Calculate all derivatives as in $\eqref{eq:As}$
\end{enumerate}

The following algorithm calculates all products of $x_i$ in step 1:

\begin{codebox}
	\Procname{$\proc{GetProducts} (x, products)$}
	\li $k = 1$
	\li $products[0] = 1$
	\li $products[1] = x[1]$
	\li \For $i = 2$ \To sizeof(x)
	\li $\{$   \Indentmore 
	\li 	$products[k+1] = x[i]$
	\li 	\For $j = 1$ \To $k$    \Indentmore 
	\li		$products[k+j+1] = x[i]*x[j]$ \End
	\li	$k = k*2 + 1$	\End 
	\li$\}$
\end{codebox}

Note that an item $i$ of $products$ in \proc{GetProducts} and the items $i$ of $q_i(\vec{x})$ and ${r_i}$ arrays in \eqref{eq:Qs} relate to this product:
	$$\prod_{\substack{k=1\\i=(i_n...i_k...i_1)_{2}}}^{n}x_{k}^{i_{k}}$$

Note also that $n+1$ items in $products$ do not require multiplications.

Getting the results in \eqref{eq:Qs} after having all the products is straightforward and requires $2^n - 1$ multiplications more.

Let us apply 'divide and conquer' algorithm to the Step 2. Split all $q_j$ for $j=(j_p \ldots j_1)_2$ into two groups $q_{{(0j_{p-1} \ldots j_1)}_2} $ and $q_{{(1j_{p-1} \ldots j_1)}_2}$. Assume that the Step 2 for ($p-1$)  is executed for both halves and the results are $\plainA$ and $\aWithTilde$. Then the following operations produce the results for $p$:

	$$ a_{{(1j_{p-1} \ldots j_1)}_2} = \aWithTilde$$
	$$ a_{{(0j_{p-1} \ldots j_1)}_2} = \plainA + \aWithTilde$$

\proc{GetPartialSums} implements this algorithm:

\begin{codebox}
	\Procname{$\proc{GetPartialSums} (q2a, xSize)$}
	\li $addendPositionsDifference$ = 1
	\li $clusterPositionsDifference$ = 2
	\li $q2aSizeMinusOne$ = sizeof($q2a$) - 1
	\li $lastSumPositionLimitInSilo = q2aSizeMinusOne$
	\li \For $iSilo$ = 1 \To $xSize$
	\li $\{$   \Indentmore 
	\li 	\For $i$ = 0 \By $clusterPositionsDifference$ \To $lastSumPositionLimitInSilo$
	\li 	$\{$   \Indentmore 
	\li 		\For $j$ = 0 \To $addendPositionsDifference$ - 1    \Indentmore 
	\li			$q2a[i+j] = q2a[i+j] + q2a[i+j+addendPositionsDifference]$	\End    \End
	\li 	$\}$
	\li	$lastSumPositionLimitInSilo = q2aSizeMinusOne - clusterPositionsDifference$
	\li	$addendPositionsDifference = clusterPositionsDifference$
	\li	$clusterPositionsDifference = 2*clusterPositionsDifference$ 	\End
	\li$\}$
\end{codebox}

It starts with $q2a$ containing all $q_i$ and finishes with $q2a$ containing all $a_i$. $xSize = n$ and $q2aSizeMinusOne = 2^n - 1$ here. Once the memory for $q_i$ is allocated this algorithm does not require any additional memory except of several scalar variables.

\paragraph{}
This figure presents the  \proc{GetPartialSums} diagrams for $n=2, 3, 4$:

 \includegraphics[width=\linewidth]{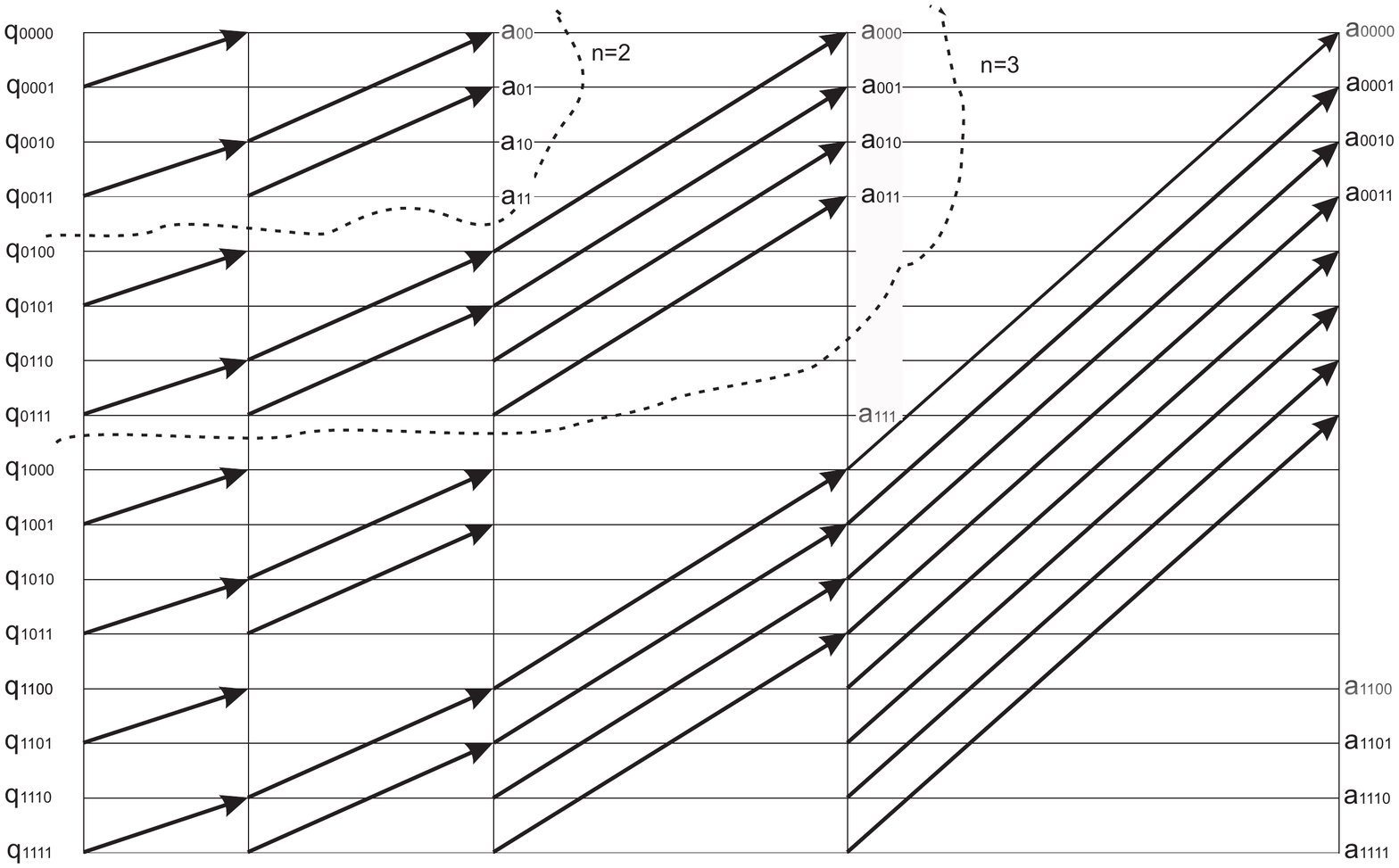}

$q_i$ input data are located along the left vertical line. $a_i$ output data are located along the right vertical line. An arrow represents an addition where addends are taken at the levels of the arrow's end points and the sum is placed into the location of the addend the arrow points to. The additions are executed silo-by-silo from the left silo to the right one. The top-left parts of the figure outlined by the broken lines represent diagrams for $n=2$ and $n=3$.

As it is proven in the next section the algorithm implemented by \proc{GetPartialSums} requires a minimal number of floating point additions. The adjective 'fast' is appropriate here then.
\theoremstyle{definition}
\begin{definition}
	The algorithm implemented by \proc{GetPartialSums} is called a Fast Multilinear Polynomial Derivatives ($FMPD$) algorithm.
\end{definition}

Having $products$ calculated in Step 1 and $q2a$ calculated in Step 2 the last Step 3 gets derivatives of MP dividing each corresponding $q2a[i]$ containing $a_i$ by associated $products[i]$.

\section {Floating Point Operation Number in Multilinear Polynomial Derivatives Calculation}
\label{Complexity}

\paragraph{}

Consider floating point operation numbers in Steps 1-3.

Step 1 requires ${(2^n - n -1)}$ multiplications at least to produce ${(2^n - n -1)}$ different products.  \proc{GetProducts} achieves this minimum because it produces one product per multiplication.
Overall Step 1 requires 
	$$ 2^n - n - 1+ 2^n - 1 = 2^{n+1} - n - 2$$
multiplications. 

Step 3 requires $(2^n - 2)$  divisions because $a_0$ and $a_{2^{n}-1} = r_{2^{n}-1}$ do not require divisions. It achieves the minimal floating point complexity because it produces the rest of $(2^n - 2)$ derivatives - one derivative per division.

Step 2 $FMPD$ requires $ n \cdot 2^{n-1}$ additions. Its asymptotic complexity is the greatest out of all Steps. The following statement  holds:

\begin{theorem}
	The $FMPD$ algorithm uses the minimal number of floating point additions for \eqref{eq:PartialSums} evaluation out of all algorithms using additions only.
\end{theorem}
\begin{proof}

	The theorem proposition about the algorithms under consideration contains the 'additions only' restriction. It is an open question if the proposition holds for algorithms using additions and subtractions.

	Let us call an algorithm to be a direct one if it uses for each $a_i$ only the addends in \eqref{eq:PartialSums} and only once each. The direct algorithms never add an $q_{0i_{k} \ldots i_0}$ item to $a_{1i_{k} \ldots i_0}$, for example. If a non-direct algorithm does it the $q_{0i_{k} \ldots i_0}$ item has to be subtracted from $a_{0i_{k} \ldots i_0}$ as well. Thus 'additions only' restriction is equivalent to 'the direct algorithms only' restriction.
	The $FMPD$ algorithm uses $n \cdot 2^{n-1}$ additions. Below we prove that this is the minimal number of additions required for \eqref{eq:As}. 
\paragraph{}
	The theorem is valid for $n=2$. Indeed, four different numbers $a_0$, $a_1$, $a_{10}$ and $a_{11}$ are produced by four additions used by $FMPD$. One needs at least one operation to produce one number. The $FMPD$'s number of additions for $n=2$ is $n \cdot 2^{n-1}=4$ and is equal to the minimal required.
\paragraph{}
Assume that the theorem in valid for $k$ and prove that then it is valid for $k+1$ as well. 
\paragraph{}
Let us split all additions into three sets:
\begin{itemize}[noitemsep]
\item those with both addends being $q_{1i_{k-1} \dots i_0}$ (lower half of $q_i$ in the diagram) or the sums of them (the 1st set);
\item those with exactly one addend being $q_{1i_{k-1} \dots i_0}$ or  the sums of them (the 2nd set);
\item the rest of additions (the 3rd set).
\end{itemize}

Note that these three sets do not overlap and their superset includes all additions used in \eqref{eq:As}.

Let us introduce  $b_j$:
	$$b_{j_{k-1} \ldots j_0} = a_{1i_{k-1} \dots i_0}$$
Any algorithm evaluating $a_{1i_{k-1} \dots i_0}$ does not use additions where the addends are $q_{0i_{k-1} \ldots i_0}$ or their sums according to  \eqref{eq:PartialSums}. Otherwise the use of subtractions is necessary to remove $q_{0i_{k-1} \ldots i_0}$ from the result. The first set execution obtains $b_j$ and according to our assumption contains $k \cdot 2^{k-1}$ additions at least. 

Consider the 2nd set of additions. Each $a_{0i_{k-1} \dots i_0}$ has to use at least one additions from the 2nd set hence there are at least $2^k$ additions because the number of $a_{0i_{k-1} \dots i_0}$ items is  $2^k$.

Let us prove that the 3rd set  has to have at least $k \cdot 2^{k-1}$ additions. If there is an algorithm $A$ for $(k+1)$ that needs lesser number of additions consider its application to $\vec{x^*} = \{0, x_k, \ldots, x_1\}$. Only additions from 3rd set remain because each addition from first two sets will have at least one $0$ addend. $A$ application to $\vec{x^*}$ reduces original target dimension from $(k+1)$ to $k$ and it can't use less than $k \cdot 2^{k-1}$ additions because it contradicts our assumption about the minimal number of additions for $k$.

Finally, we have minimal numbers of additions for all three sets. These sets do not overlap and their superset is a set of all additions used for \eqref{eq:PartialSums} evaluation. Thus the total number of additions in the three sets is:
	$$A_{k+1} \geq k \cdot 2^{k-1} + k \cdot 2^{k-1} + 2^k = (k+1) \cdot 2^k$$

The right hand side of the above is exactly the number of additions used by $FMPD$ for $k+1$ meaning that it uses the minimal number of the required additions.
\end{proof}

An $FMPD_l$ analog of $FMPD$ for derivatives up to $l$th order ($l<n$) can be produced by omitting in $FMPD$ the additions that do not contribute to the required partial sums. In the last silo the omitted additions will be those which arrows point to $a_i$ that have the number of 1s in the binary presentation of $i$ greater than $l$. In the last but one silo the omitted arrows in the top cluster will be those pointing to $a_i$ with the number of 1s in $i$ greater than $l+1$. The lower cluster repeats the the mission pattern of the top cluster.
For an $iSilo$ the omitted arrows in the top cluster will point to $a_i$ with the number of 1s in $i$ greater than $l + 1$. The rest of cluster in this cluster repeat the mission pattern of the first one.
The first silo with omissions is $iSilo_{first} = l+2$.
This figure presents the $FMPD_l$ diagram for $n=5, l=2$ (the omitted additions are presented with the broken line):
	\includegraphics[width=\linewidth]{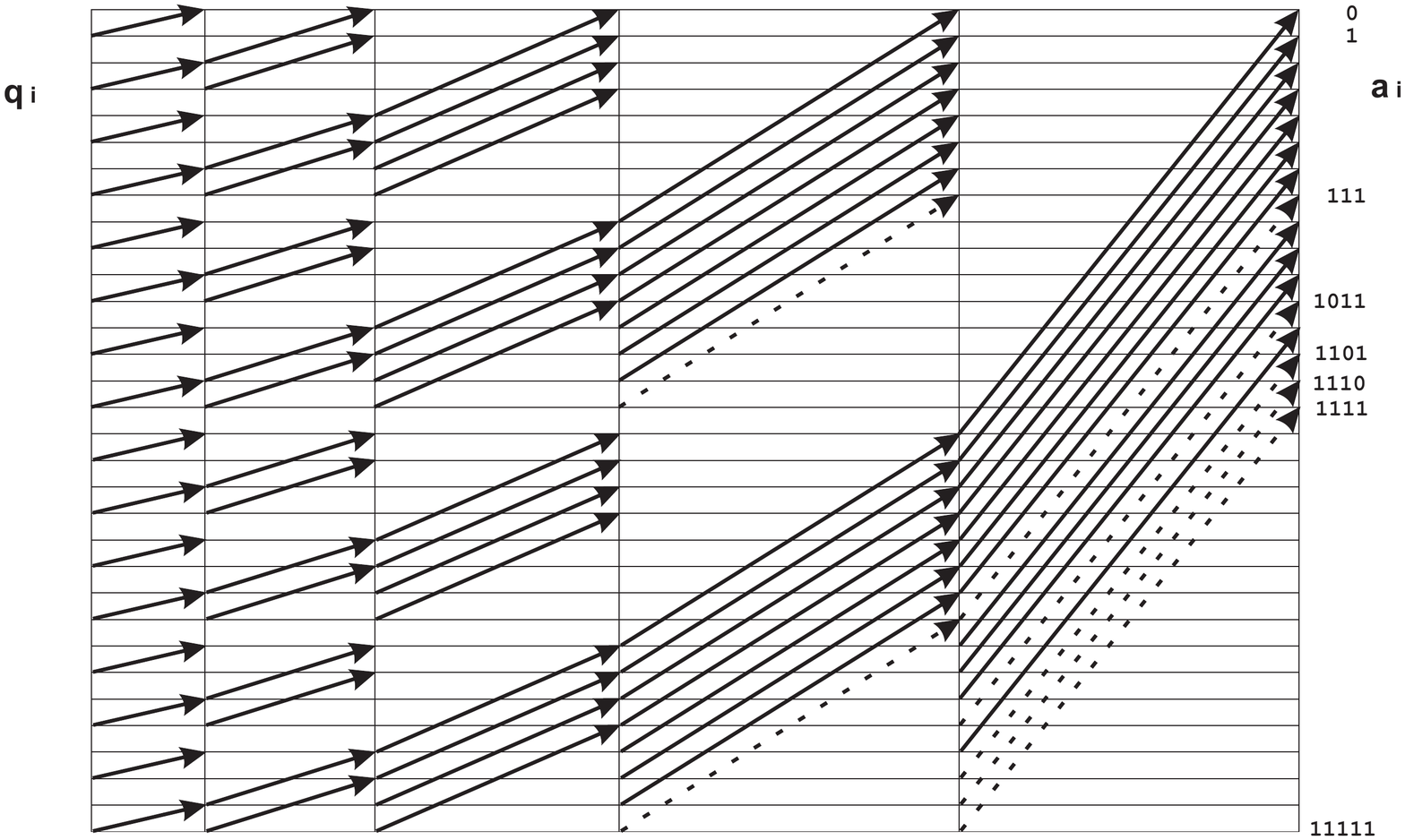}

Note that the above omission can be precalculated for the usage in the internal loop of \proc{GetPartialSums} and have a negligible effect on the overall execution time.

The number of floating point additions for \eqref{eq:PartialSums} evaluation by $FMPD_l$ \cite{avtoreferat}:
\begin{equation} \label{eq:AFMPDl}
	A_{FMPD_l} = 
	\begin{cases}
		(l+1) \cdot 2^{n-1} + \sum\limits_{k=l+2}^n \bigg[ \sum\limits_{i=0}^l  \binom{k-1}{i} \bigg] \cdot 2^{n-k}
					& \text{if $l < n-1$},\\
		n \cdot 2^{n-1} 	& \text{if $l = n-1$}.
	\end{cases}
\end{equation}

An important particular case of \eqref{eq:AFMPDl} is the evaluation of the polynomial and its gradient. Let us estimate the relative increase in  the number of additions required for obtaining the gradient in comparison with the MP evaluation itself:
	$$  R_1(n) = [A{FMPD_1} - A_0(n)]/A_0(n) = [2\cdot 2^{n-1} + \sum \limits_{k=3}^{n} k \cdot 2^{n-k} - A_0(n)]/A_0(n) $$
where $ A_0(n) = \sum \limits_{k=1}^n 2^{n-k} $.

Using $ \sum\limits_{i=1}^{\infty}  2^{-i} = 1 $ we get
	$$ \lim_{n\to\infty} A_0(n) = 2^n $$ 
and using $ \sum\limits_{i=1}^{\infty} i \cdot 2^{-i} = 2 $ \cite{Quora1} we get
	$$ \lim_{n\to\infty} R_1(n) = 1/8 $$

Thus the cost of evaluating all partial sums for first derivatives approaches to only 1/8 of MP evaluation cost with a growing number of variables. The cost per first derivative is negligible in comparison with MP evaluation.

Similarly, using $ \sum\limits_{i=1}^{\infty} i^2 \cdot 2^{-i} = 6 $ \cite{Quora2} we get
	$$ \lim_{n\to\infty} R_2(n) = 1 $$
Thus the cost of $n(n+1)/2$ first and second derivatives evaluation approaches the cost of MP evaluation itself. The cost per derivative is negligible in comparison with MP evaluation.

The number of additions in the algorithm calculating partial sums in \eqref{eq:As} for each $i$ separately  is:
	$$ A_{naive} =\sumOfCombinations \cdot (2^{n-i} - 1) = 2^n \cdot \sumOfCombinations \cdot 2^{-i} - 2^n = 3^n - 2^n$$

The particular variant of the binomial formula was used here:
	$$\sumOfCombinations \cdot u^i = (1+u)^n$$

The number of additions in the algorithm calculating partial sums in \eqref{eq:As} for each $i$ separately for the up to and including $l$ derivatives is:
	$$ A_{naive_l} = \sum\limits_{i=0}^{l} \binom{n}{i} \cdot (2^{n-i} -1) $$

The number of additions in the algorithm calculating partial sums in \eqref{eq:As} for each $i$ separately for the up to and including $2$ derivatives is:
	$$ A_{naive_2} = 1 + n \cdot 2^{n-1} + \binom{n}{2} \cdot 2^{n-2} = 1+n \cdot 2^{n-1} + 2^{n-2} \cdot n!/[2! \cdot (n-2)!] = 1+ n \cdot 2^{n-1} + n^2 \cdot 2^{n-3} - n \cdot 2^{n-3}$$

The asymptotic complexity expressions of these algorithms are:
\begin{align*}
	& O_{FMPD}(n) = n \cdot 2^n \\
	& O_{FMPD_{1}}(n) = 2^n \\
	& O_{naive_{n}}(n) =3^n \\
	& O_{naive_{2}}(n) = n^2 \cdot 2^n
\end{align*}

$FMPD$ has a substantially lower complexity than the naive algorithm. The exponential complexity of $FMPD$ and $FMPD_l$ allow their practical usage only for relatively low $n$ values. Still it is noticeably faster than its naive counterpart. For example, for $n=8$ $FMPD$ uses 1024 additions vs. 6305 used by the naive algorithm. The table below compares number of additions in $FMPD$ and $FMPD_2$ against their naive counterparts:

\begin{center}
\begin{tabular}{|c|c|c|c|c|c|c|c|c|c|c|} 
\hline
	$\boldsymbol{n}$ & \textbf{2} & \textbf{3} & \textbf{4} & \textbf{5} & \textbf{6} & \textbf{7} & \textbf{8} & \textbf{9} & \textbf{10} & $O(n)$ \\
\hline
	$\boldsymbol{A_{naive_{n}}/A_{FMPD}}$ & 1.25 & 1.58 & 2.0  & 2.6 & 3.5 & 4.6 & 6.2 & 8.3 & 11.3 & $1.5^n$\\ 
\hline
	$\boldsymbol{A_{naive_{2}}/A_{FMPD_{2}}}$ &  1.25 & 1.58 & 2.0  & 2.4 & 2.8 & 3.5 & 4.2 & 5.0 & 5.5 & $n$ \\ 
\hline
\end{tabular}
\end{center}

Now consider parallelization opportunities for the floating point operations in Steps 1-3.

Step 1. The body of the internal loop of \proc{GetProducts} can be split across $m$ processors. If a single processor asymptotic complexity is $O_{1}(n)$ then the one for $m$ processors is $O_{m}(n)=O_{1}(n)/m$. 
After the products in \eqref{eq:Qs} are evaluated the multiplications by $r_{i}({\vec{r})}$ can be done in $M_{m}(n) = \lfloor M_{1}(n)/m \rfloor + 1$ operations per processor, where $M_{1}(n)$ and $M_{m}(n)$ are the number of multiplications for a single processor and each of $m$ processors.

Step 2. As it is illustrated by $FMPD$ diagram all additions in a silo can be executed in parallel because each addition has its own separate addends and the result. Then $A_{m}(n) = \lfloor A_{1}(n)/m \rfloor + 1$, where $A_{1}(n)$ and $A_{m}(n)$ are the number of additions for a single processor and each of $m$ processors.

Step 3. All multiplications may be executed in parallel. It requires $M_{m}(n)=\lfloor M_{1}(n)/m \rfloor + 1$ multiplications.

\section {Fast Multilinear Polynomial Derivatives and FFT}

The simplified $FFT$ algorithm \cite{Aho} has the following diagram:

 \includegraphics[width=\linewidth]{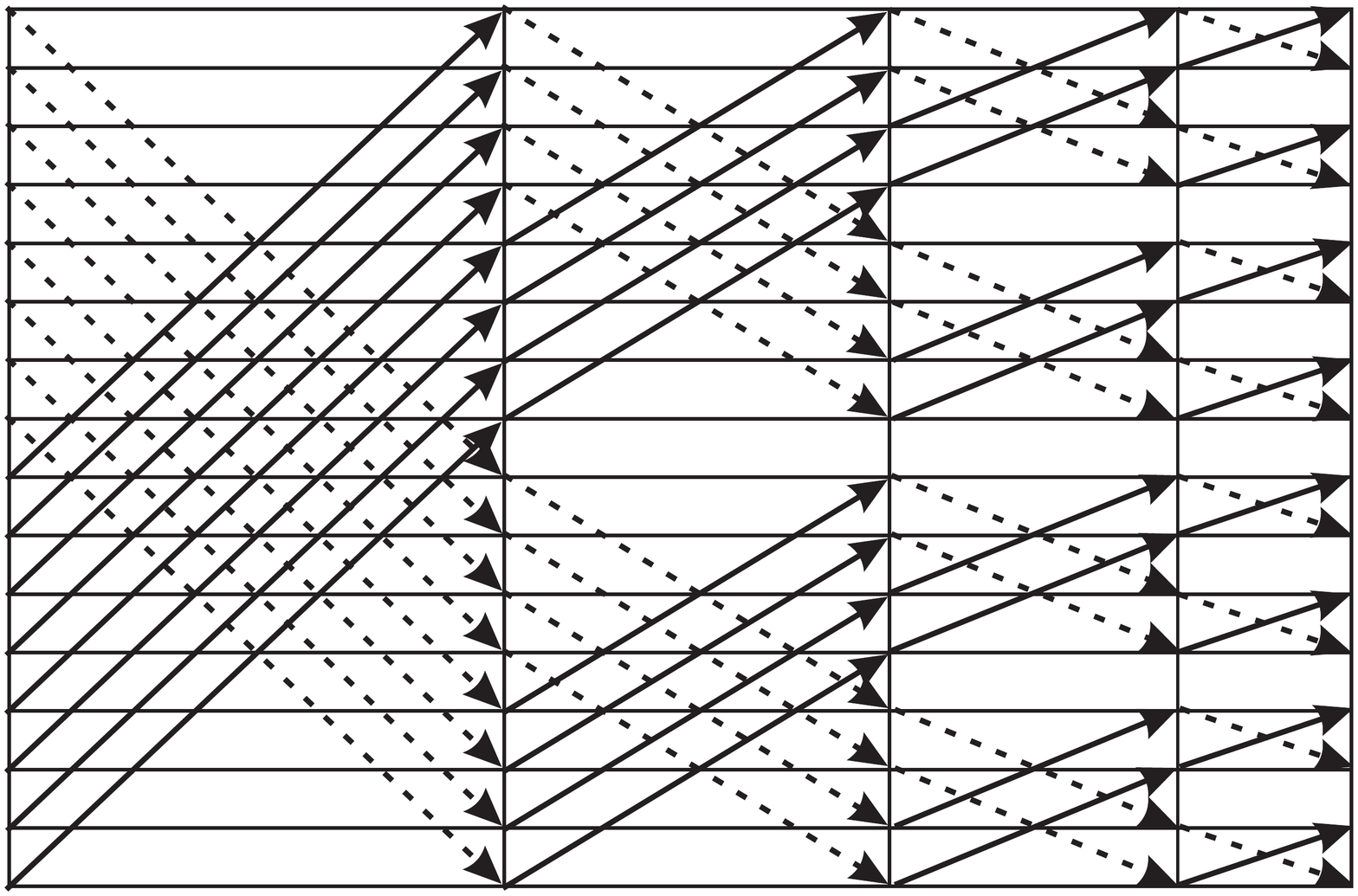}
where the solid line arrows mean the same as ones in $FMPD$ diagram and the broken line arrows turn into void operations if $\omega =0$. Comparison of $FFT$ ($\omega =0$) with the $FMPD$ diagram reveals their structural similarity. Note that a somewhat simpler MP derivatives problem allowed to achieve the minimal possible number of operations whereas only a complexity lower bound for FFT \cite{fftComplexity} is established.

\section*{Acknowledgements}
The author would like to thank Joris van der Hoeven and Igor Shparlinski for constructive criticism of the article and advise.

\medskip

\raggedright

\end{document}